\documentclass[12pt]{article}
\usepackage{amssymb,amsmath,amsthm,amscd,latexsym}
\usepackage{mathrsfs}
\usepackage{mathrsfs}
\usepackage{amsfonts}
\usepackage{amsmath}
\usepackage{amssymb}
\usepackage{arydshln}
\usepackage{amscd}
\usepackage{cite}
\usepackage[colorlinks,linkcolor=red,anchorcolor=blue,citecolor=green]{hyperref}
\allowdisplaybreaks[3]

\renewcommand{\paragraph}{\roman{paragraph}}
\input cyracc.def

\def\F{\mathbb{F}}

\def\F{\mathbb{F}}

\newtheorem{theorem}{Theorem}
\newtheorem{corollary}{Corollary}
\newtheorem{example}{Example}

\newtheorem{lemma}{Lemma}

\theoremstyle{definition}
\newtheorem{definition}{Definition}
\newtheorem{remark}{Remark}

\begin{document}
\title{Two classes of LCD codes derived from $(\mathcal{L},\mathcal{P})$-TGRS codes}

\author{Ziwei Zhao$^{1}$, Xiaoni Du$^{*1,2,3}$, Xingbin Qiao$^{1}$\\
1. College of Mathematics and Statistic, \\Northwest Normal University, Lanzhou, 730070, China.\\
2. Key Laboratory of Cryptography and Data Analytics, \\Northwest Normal University,\\ Lanzhou, 730070, China.\\
3. Gansu Provincial Research Center for Basic Disciplines\\ of Mathematics and Statistics,\\ Northwest Normal University, Lanzhou, 730070, China.
}
\date{}
\maketitle
{\bf Abstract~~} Twisted generalized Reed-Solomon (TGRS) codes, as a flexible extension of classical generalized Reed-Solomon (GRS) codes, have attracted significant attention in recent years. In this paper, we construct two classes of LCD codes from the $(\mathcal{L},\mathcal{P})$-TGRS code $\mathcal{C}_h$ of length $n$ and dimension $k$, where $\mathcal{L}=\{0,1,\ldots,l\}$ for $l\leq n-k-1$ and $\mathcal{P}=\{h\}$ for $1\leq h\leq k-1$. First, we derive the parity check matrix of $\mathcal{C}_h$ and provide a necessary and sufficient condition for $\mathcal{C}_h$ to be an AMDS code. Then, we construct two classes of LCD codes from $\mathcal{C}_h$ by suitably choosing the evaluation points together with certain restrictions on the coefficient of $x^{h-1}$ in the polynomial associated with the twisting term. From the constructed LCD codes we further obtain two classes of LCD MDS codes. Finally, several examples are presented.

{\bf Key words~~}parity check matrix, LCD code, AMDS code, LCD MDS code, twisted generalized Reed-Solomon code

\section{Introduction}
Let $\mathbb{F}_{q}$ be the finite field of $q$ elements, where $q=p^{m}$ with prime $p$ and positive integer $m$. An $[n,k,d]$ linear code $\mathcal{C}$ over $\mathbb{F}_{q}$ is a $k$-dimensional linear subspace of $\mathbb{F}_{q}^{n}$ with minimum (Hamming) distance $d$. The Euclidean dual code of the linear code $\mathcal{C}$ is denoted by $$\mathcal{C}^{\perp}=\left\{\boldsymbol{y}\in\mathbb{F}_q^n:
\langle\boldsymbol{x},\boldsymbol{y}\rangle=0,\text{ for all }\boldsymbol{x}\in\mathcal{C}\right\},$$
where $\langle\boldsymbol{x},\boldsymbol{y}\rangle=\sum_{i=1}^{n}x_{i}y_{i}$ for vectors
$\boldsymbol{x}=(x_{1},x_{2},\ldots,x_{n})$ and $\boldsymbol{y}=(y_{1},y_{1},\ldots,y_{n})$ in $\mathbb{F}_{q}^{n}$. A generator matrix of an $[n,k]$ code $\mathcal{C}$ over $\mathbb{F}_{q}$ is any $k\times n$ matrix $G$ whose rows form a basis of $\mathcal{C}$. In particular, there is an $(n-k)\times n$ matrix $H$, called a parity check matrix for the $[n,k]$ code $\mathcal{C}$, defined by $\mathcal{C}=\{\mathbf{x}\in\mathbb{F}_{q}^{n}:H\boldsymbol{x}^{T}=0\}$.

For an \([n, k, d]\) linear code \(\mathcal{C}\), the parameters \(n\), \(k\) and \(d\) must satisfy the Singleton bound [\ref{t42}]: $d \leq n - k + 1$. The code \(\mathcal{C}\) is called a maximum distance separable (MDS) code if \(d = n - k + 1\), and almost MDS (AMDS) code if $d=n-k$.
MDS and AMDS codes have numerous applications in communication and storage systems, such as data recovery, cybersecurity, secret sharing and design constructions [\ref{t1}-\ref{t7}]. Another important class of codes is linear complementary dual (LCD) codes, the code with $\mathcal{C}\cap\mathcal{C}^{\perp}=\{\boldsymbol0\}$, which was introduced by Massey in 1992 [\ref{t8}]. Bringer et al. [\ref{t9}] and Carlet et al. [\ref{t10}] found that binary LCD codes play a decisive role in implementations against side-channel and fault injection attacks. Due to their strong algebraic structures and interesting practical applications, numerous LCD, MDS and LCD MDS codes have been constructed using various methods [\ref{t18},\ref{t14}]. In particular, most known constructions are based on generalized Reed Solomon codes (GRS) codes and their extensions since GRS codes are all MDS [\ref{t19}].

As an extension of GRS codes, twisted GRS codes are firstly introduced by Beelen et al. [\ref{t16}] in 2017 inspired by the construction for twisted Gabidulin codes. Unlike GRS codes, TGRS codes may not be MDS codes. By choosing appropriate twists, many researchers have investigated TGRS codes and constructed various classes of codes, including MDS codes [\ref{t18}-\ref{t20}] that are not equivalent to GRS codes, NMDS codes [\ref{t21}], self dual codes [\ref{t22}], LCD codes [\ref{t13}-\ref{t14}] and LCD MDS codes [\ref{t24}]. Moreover, it was also shown that TRS codes possess favorable structural properties, making them suitable as alternatives to Goppa codes in the McEliece code-based cryptosystem [\ref{t25}].

For $(\mathcal{L},\mathcal{P})$-TGRS codes (see Definition \ref{TGRS} for more details), in 2017, Beelen et al. [\ref{t16}] provided a sufficient and necessary condition for the 1-TGRS codes ($\mathcal{L}=\{1\},\mathcal{P}=\{0\}$ or $\{k-1\}$) to be MDS. Further related results on 1-TGRS codes can be found in [\ref{t21},\ref{t24},\ref{t27}]. Subsequently, 2-TGRS codes [\ref{t28}-\ref{t30}], 3-TGRS codes [\ref{t32}], $l$-TGRS ($|\mathcal{L}|=|\mathcal{P}|=l$) codes [\ref{t33}-\ref{t35}] also attracted considerable attention. 
More recently, Zhao et al. [\ref{t31}] established a necessary and
sufficient condition for general $(\mathcal{L},\mathcal{P})$-TGRS codes to be MDS, and later Hu et al. [\ref{t37}] proposed a more precise definition of $(\mathcal{L},\mathcal{P})$-TGRS codes, and a necessary and sufficient condition for them to be MDS codes. On the other hand, several constructions of LCD and LCD MDS codes based on TGRS codes have been reported. In particular, three classes of LCD MDS codes [\ref{t18}] and a class of LCD codes [\ref{t20}] were constructed based on 1-TGRS codes ($\mathcal{L}=\mathcal{P}=\{0\}$ ). Most recently,
several classes of LCD MDS codes [\ref{t30}] and four classes of LCD codes [\ref{t41}] were constructed, respectively, based on $(\mathcal{L},\mathcal{P})$-TGRS codes with $\mathcal{L}=\{0,1\}$ or $\{0,2\}$, $\mathcal{P}=\{0\}$, and $(\mathcal{L},\mathcal{P})$-TGRS codes with $\mathcal{L}=\{0,1,\ldots,l\}$, $\mathcal{P}=\{0\}$.


Motivated by the above works, in this paper we focus on  two classes of LCD codes from the $(\mathcal{L},\mathcal{P})$-TGRS code $\mathcal{C}_h$ with $\mathcal{L}=\{0,1,\ldots,l\}$ for $l\leq n-k-1$ and $\mathcal{P}=\{h\}$ for $1\leq h\leq k-2$. First, we derive an explicit parity check matrix of $\mathcal{C}_h$ and investigate its AMDS property. Then, we construct two classes of LCD codes from $\mathcal{C}_h$ by suitably choosing the evaluation points together with certain restrictions on the coefficient of $x^{h-1}$ in the polynomial associated with the twisting term. In particular, our constructions include two classes of LCD MDS codes. Finally, several examples are provided to illustrate the validity of our results.

The organization of this paper is as follows. In Section 2, we introduce the necessary notation and definitions for TGRS, MDS, AMDS and LCD codes, together with several lemmas on complete symmetric polynomials and determinants. In Section 3, we present the parity check matrix of the $(\mathcal{L},\mathcal{P})$-TGRS code $\mathcal{C}_h$. In Section 4, we give a necessary and sufficient condition for $\mathcal{C}_h$ to be an AMDS code. In Section 5, we construct two classes of LCD codes from $\mathcal{C}_h$, which contain two classes of LCD MDS codes. Finally, Section 6 concludes the paper.

\section{Preliminaries}
In this section, we introduce some notation, definitions, and lemmas that will be used in subsequent sections. For convenience, we begin this section by fixing the following notation.
\begin{itemize}
  \item $q$ is a prime power, and $n,k$ are non-negative integers with $k\leq n$.
  \item $\mathbb{F}_q$ denotes the finite field of order $q$, and $\mathbb{F}_q^{\ast}=\mathbb{F}_q\setminus\{0\}$.
  \item $\mathbb{F}_q^{n}$ denotes the $n$-dimensional vector space over $\mathbb{F}_q$, and $\mathbb{F}_q^{m\times n}$ denotes the set of $m\times n$ matrices over $\mathbb{F}_q$.
  \item $|S|$ denotes the cardinality of a set $S$.
  \item $\boldsymbol{O}$ denotes the zero matrix, and $I_k$ denotes the $k\times k$ identity matrix.
  \item $A^{T}$ denotes the transpose of a matrix $A$.
  \item $\mathrm{Coeff}_{m}(a(x))$ denotes the coefficient of $x^{m}$ in the polynomial $a(x)=\sum_{i=0}^{n} a_i x^i$ with $n\geq m$, that is, $\mathrm{Coeff}_{m}(a(x))=a_m$.

\end{itemize}
\subsection{GRS and TGRS codes}
Let $\mathbb{F}_q[x]$ be the polynomial ring over $\mathbb{F}_q$. Let  $\boldsymbol{\alpha}=(\alpha_1,\alpha_2,\ldots,\alpha_n)\in\mathbb{F}_q^{n}$ and $\boldsymbol{v}=(v_1,v_2,\ldots,v_n)\in(\mathbb{F}_q^{\ast})^{n}$, where $\alpha_1,\alpha_2,\ldots,\alpha_n$ are pairwise distinct elements. The generalized Reed-Solomon (GRS) code $\mathcal{C}(\boldsymbol{\alpha},\boldsymbol{v})$ with length $n$ and dimension $k$ is defined by
$$\mathcal{C}(\boldsymbol{\alpha},\boldsymbol{v})
=\{\bigl(v_1f(\alpha_1),\ldots,v_nf(\alpha_n)\bigr)\,|\,f(x)\in\mathbb{F}_q[x], \deg(f(x))\leq k-1\}.$$
When $\boldsymbol{v}=(1,1,\ldots,1)$, the code is referred to as the Reed-Solomon (RS) code.

TGRS codes are an extension of GRS codes by adding certain monomials (called twists) to specific positions (called hooks) in the generating polynomial $f(x)\in \mathbb{F}_q[x]$ of GRS codes. Hu et al. [\ref{t37}] provided a more precise definition of TGRS codes, which is stated as follows.

\begin{definition}$[\ref{t37}]$\label{TGRS}
Let $n,k$ and $l$ be integers such that $0\le k\le n$ and $0\le l\le n-k$. Let $\mathcal{L}\subseteq\{0,1,\ldots,n-k-1\}$ (called the twist set) with $|\mathcal{L}|=l+1$,
$\mathcal{P}\subseteq\{0,1,\ldots,k-1\}$ (called the position set), and let
$B=(b_{i,j})\in\mathbb{F}_q^{k\times(n-k)}$ satisfy $b_{i,j}=0$ whenever $i\notin \mathcal{P}$ or $j\notin \mathcal{L}$.
Define the twisted polynomial space
\begin{align}\label{F}
\mathcal{F}(\mathcal{L},\mathcal{P},B)
=
\Big\{
f(x)=\sum_{i=0}^{k-1} f_i x^i
+
\sum_{i\in \mathcal{P}} f_i\sum_{j\in \mathcal{L}} b_{i,j}\,x^{k+j}
|\ f_0,\ldots,f_{k-1}\in\mathbb{F}_q
\Big\}.
\end{align}
Then for any pairwise distinct evaluation points $\boldsymbol{\alpha}=(\alpha_1,\alpha_2,\ldots,\alpha_n)\in\mathbb{F}_q^n$ and $\boldsymbol{v}=(v_1,v_2,\ldots,v_n)\in(\mathbb{F}_q^\ast)^n$, the $(\mathcal{L},\mathcal{P})$-twisted generalized Reed-Solomon (TGRS) code is the evaluation code
$$
\mathcal{C}(\mathcal{L},\mathcal{P},B)
=
\Big\{
\big(v_1 f(\alpha_1),\ldots, v_n f(\alpha_n)\big)
|\ f\in\mathcal{F}(\mathcal{L},\mathcal{P},B)
\Big\}
\subseteq \mathbb{F}_q^n.
$$
Specifically, when $\boldsymbol{v}=(1,1,\ldots,1)$, the code is called $(\mathcal{L},\mathcal{P})$-TRS code, and when $B=\boldsymbol{O}$, the code  $\mathcal{C}(\mathcal{L},\mathcal{P},B)$ reduces to the usual GRS code $\mathcal{C}(\boldsymbol{\alpha},\boldsymbol{v})$.
\end{definition}

By definition, the properties of an $(\mathcal{L},\mathcal{P})$-TGRS code depend not only on $\boldsymbol{\alpha}$ and $\boldsymbol{v}$, but also on the choices of $\mathcal{L},\mathcal{P}$. In  [\ref{t41},\ref{t47}], Liang et al. considered $\mathcal{L}=\{0,1,\ldots,l\}$ for $l\leq n-k-1$, $\mathcal{P}=\{h\}$ for $h=0$ or $h=k-1$. It is worth noting that [\ref{t41}] investigated four classes of LCD codes for the case $h=0$, while [\ref{t47}] did not discuss the LCD property for the case $h=k-1$. Therefore, it is natural to examine what will happen in the case $1\leq h\leq k-1$.

Below in this paper, we always assume that $\mathcal{L}=\{0,1,\ldots,l\}$ for $l\leq n-k-1$,  $\mathcal{P}=\{h\}$ for $1\leq h\leq k-1$, and $\boldsymbol{\eta}=(\eta_0,\eta_1,\ldots,\eta_l)\in \mathbb{F}_{q}^{l+1}$. Then Eq.\eqref{F} can be written as
$$
\mathcal{F}(\mathcal{L},\{h\},\boldsymbol{\eta})
=
\Big\{
f(x)=\sum_{i=0}^{k-1} f_i x^i
+f_h\sum_{t\in \mathcal{L}} \eta_{t}\,x^{k+t}
\,|\ f_i\in\mathbb{F}_q
\Big\},
$$
and the corresponding $(\mathcal{L},\mathcal{P})$-TGRS code $\mathcal{C}(\mathcal{L},\{h\},\boldsymbol{\eta})$ are defined by
\begin{equation}\label{Ch}
\mathcal{C}(\mathcal{L},\{h\},\boldsymbol{\eta})
=
\Big\{
\big(v_1 f(\alpha_1),\ldots, v_n f(\alpha_n)\big)
\,|\ f\in\mathcal{F}(\mathcal{L},\{h\},\boldsymbol{\eta})
\Big\},
\end{equation}
denoted by $\mathcal{C}_{h}$ for short. Obviously, $\mathcal{C}_{h}$ has a generator matrix
\begin{equation}\label{G}
G_h=\left(
    \begin{array}{ccc}
      v_1 & \cdots & v_n \\
      \vdots & \ddots & \vdots \\
      v_1\alpha_1^{h-1} & \cdots & v_n\alpha_n^{h-1} \\
      v_1(\alpha_1^{h}+\sum_{t\in \mathcal{L}} \eta_{t}\,\alpha_1^{k+t}) &\cdots & v_n(\alpha_n^{h}+\sum_{t\in \mathcal{L}} \eta_{t}\,\alpha_n^{k+t}) \\
      v_1\alpha_1^{h+1} & \cdots & v_n\alpha_n^{h+1} \\
      \vdots & \ddots & \vdots \\
      v_1\alpha_1^{k-1} & \cdots & v_n\alpha_n^{k-1} \\
    \end{array}
  \right).
\end{equation}

\subsection{Fundamentals of linear code on MDS, AMDS and LCD}

In this subsection, we present several lemmas that give equivalent characterizations of MDS and AMDS codes.
\begin{lemma}\label{MDS}$[\ref{t42}]$
Let $\mathcal{C}$ be an $[n,k]$ linear code over $\mathbb{F}_q$, and let $G$ be a generator matrix of $\mathcal{C}$. Then $\mathcal{C}$ is an MDS code if and only if the determinant of any $k\times k$ submatrix of $G$ is nonzero.
\end{lemma}

It is well known that the dual of an MDS code is also an MDS code.

\begin{lemma}\label{AMDS}$[\ref{t43}]$
An $[n,k]$ linear code $\mathcal{C}$ over $\mathbb{F}_q$ is AMDS if and only if
a generator matrix $G$ of $\mathcal{C}$ satisfies the following conditions:
\begin{enumerate}
  \item  [$(1)$] There exist $k$ linearly dependent columns in $G$;
  \item  [$(2)$] Any $k+1$ columns of $G$ have rank $k$.
\end{enumerate}
\end{lemma}

In what follows, we show a result related to the hull of linear codes and LCD codes.

Let $\mathcal{C}\subseteq \mathbb{F}_q^{n}$ be a linear code. The Euclidean hull of $\mathcal{C}$ is defined as
\[
Hull(\mathcal{C})=\mathcal{C}\cap \mathcal{C}^{\perp}.
\]
If $Hull(\mathcal{C})=\{\boldsymbol0\}$, then $\mathcal{C}$ is called an Euclidean linear complementary dual (LCD) code. The result below follows directly from the definition.

\begin{lemma}\label{LCD}$[\ref{t30}]$
Let $\mathcal{C}$ be an $[n,k]$ linear code over $\mathbb{F}_q$. Suppose that $G_k$ and $H_{n-k}$ are the generator matrix and parity check matrix of $\mathcal{C}$, respectively. Then
$\dim(\operatorname{Hull}(\mathcal{C}))=n-\operatorname{rank}\!\begin{pmatrix} G_k\\ H_{n-k}\end{pmatrix}$. Moreover, $\mathcal{C}$ is an LCD code if and only if
$\operatorname{rank}\!\begin{pmatrix} G_k\\ H_{n-k}\end{pmatrix}=n$.
\end{lemma}

\subsection{Some results on complete symmetric polynomials and determinants}

In this subsection, we recall the definition of the complete symmetric polynomial and some useful results on determinants.

\begin{definition}\label{matrix}$[\ref{t44},\ref{t45}]$
Let $S_t(x_1,x_2,\ldots,x_n)$ be the complete symmetric polynomial of degree $t$ defined by
\[
S_t(x_1,x_2,\cdots,x_n)=
\begin{cases}
0, & \text{\emph{if} } t<0;\\[1mm]
\displaystyle\sum_{\substack{t_1+t_2+\cdots+t_n=t,\\ t_i\ge 0,\,1\leq i\leq n}}
x_1^{t_1}x_2^{t_2}\cdots x_n^{t_n}, & \text{\emph{if} } t\ge 0,
\end{cases}
\]
and denote $S_t(x_1,x_2,\cdots,x_n)$ by $S_t$. It is clear that $S_0=0$. The generating function for the $S_t$ is
\[
\sum_{t\ge 0} S_t(x_1,x_2,\ldots,x_n)z^t=\prod_{i=1}^n\frac{1}{1-x_i z}.
\]
\end{definition}


\begin{lemma}\label{ui}$[\ref{t44}]$
Let $\alpha_1,\alpha_2,\ldots,\alpha_n\in \mathbb{F}_q$ be pairwise distinct. Let $u_i=\displaystyle\prod_{\substack{j=1, j\ne i}}^{n}(\alpha_i-\alpha_j)^{-1}$ for $1\le i\le n$.
Then for $n\ge 3$, we have
\[
\sum_{i=1}^{n}u_i\alpha_i^{h}=
\begin{cases}
0, & \text{if }~0\le h\le n-2;\\[1mm]
S_{h-n+1}(\alpha_1,\alpha_2,\cdots,\alpha_n), & \text{if }~h\ge n-1.
\end{cases}
\]
\end{lemma}

\begin{lemma}\label{wi}$[\ref{t18}]$
Let $\alpha_1,\alpha_2,\ldots,\alpha_n\in \mathbb{F}_q$ be pairwise distinct and $G=\bigl(\alpha_{i}^{j-1}\bigr)_{1\le i,j\le n}$ be the $n\times n$ Vandermonde matrix, i.e.,
\begin{align}\label{G1}
G=\begin{pmatrix}
1&\cdots&1\\
\alpha_1&\cdots&\alpha_n\\
\vdots&&\vdots\\
\alpha_1^{\,n-1}&\cdots&\alpha_n^{\,n-1}
\end{pmatrix}.
\end{align}
For an integer $r$ with $0\le r\le n-1$, let $\boldsymbol{e}_{r+1}\in \mathbb{F}_q^{n}$ be the unit vector whose $(r+1)$-th entry equals $1$ and all other entries equal $0$.
Then the system of equations over $\F_q$
\[
G(w_1,w_2,\ldots,w_n)^{T}=\boldsymbol{e}_{r+1}
\]
has a unique solution $(w_1,\ldots,w_n)^{T}$ given by
\[
w_i=(-1)^{n+r+1}u_i\,\sigma_{n-1-r}(i),\qquad 1\le i\le n,
\]
where
$$u_i=\displaystyle\prod_{\substack{j=1,~ j\ne i}}^{n}(\alpha_i-\alpha_j)^{-1}
\quad \mathrm{and} \quad
\sigma_{n-1-r}(i)=\sum_{\{i_1,\ldots,i_{n-1-r}\}\subseteq \{1,\ldots,n\}\setminus\{i\}}
\ \prod_{s=1}^{n-1-r}\alpha_{i_s}.
$$
\end{lemma}

\begin{lemma}$[\ref{t46}]$\label{l6}
Let $A$ be an invertible matrix, and let $\boldsymbol{u}, \boldsymbol{v}$ be column vectors. Then
$$
\det(A+\boldsymbol{u}\boldsymbol{v}^{T})=\det(A)\bigl(1+\boldsymbol{v}^{ T}A^{-1}\boldsymbol{u}\bigr).
$$
\end{lemma}

\section{The parity check matrix of $\mathcal{C}_{h}$}
In this section, to gain a deeper understanding of $\mathcal{C}_h$ and facilitate the arguments in the subsequent sections, we present an explicit characterization of a parity check matrix of $\mathcal{C}_h$ as follows.

\begin{theorem}\label{th1}
Let $\mathcal{C}_{h}$ and $G_{h}$ be defined by Eq.\eqref{Ch} and Eq.\eqref{G}, respectively. For $1\leq i\leq n$, define $\widetilde{w}_i=\frac{w_i}{u_i}$, where $u_i$ and $\sigma_{n-1-h}(i)$ are defined in Lemma $\ref{wi}$, respectively.
Then
\begin{equation}\label{H}
\footnotesize
H_{h}=
\begin{pmatrix}
\dfrac{u_1}{v_1}  & \cdots & \dfrac{u_n}{v_n}\\[2mm]
\dfrac{u_1}{v_1}\alpha_1  & \cdots & \dfrac{u_n}{v_n}\alpha_n\\
\vdots  & \ddots & \vdots\\
\dfrac{u_1}{v_1}\alpha_1^{\,n-k-(l+2)}  &
\cdots &
\dfrac{u_n}{v_n}\alpha_n^{\,n-k-(l+2)}\\[2mm]
\dfrac{u_1}{v_1}\!\Big(\alpha_1^{\,n-k-(l+1)}-\widetilde{w}_1\Psi_l\Big)  &
\cdots &
\dfrac{u_n}{v_n}\!\Big(\alpha_n^{\,n-k-(l+1)}- \widetilde{w}_n\Psi_l\Big)\\
\vdots  & \ddots & \vdots\\
\dfrac{u_1}{v_1}\!\Big(\alpha_1^{\,n-k-2}-\widetilde{w}_1\Psi_1\Big) &
\cdots &
\dfrac{u_n}{v_n}\!\Big(\alpha_n^{\,n-k-2}-\widetilde{w}_n\Psi_1\Big)\\[1em]
\dfrac{u_1}{v_1}\!\Big(\alpha_1^{\,n-k-1}-\widetilde{w}_1\Psi_0\Big) &
\cdots &
\dfrac{u_n}{v_n}\!\Big(\alpha_n^{\,n-k-1}-\widetilde{w}_n\Psi_0\Big)
\end{pmatrix}_{(n-k)\times n}
\end{equation}
is a parity check matrix of $\mathcal{C}_{h}$, where $\Psi_s:=\sum_{t=s}^{l}\eta_t\,S_{t-s}$ for $0\le s\le l$.

\end{theorem}

\begin{proof}
From the definition of $G$ in Eq.\eqref{G1}, together with Lemmas \ref{ui} and \ref{wi}, we have
\begin{equation}\label{eq7}
G(u_1,u_2,\ldots,u_n)^{T}=(0,\ldots,0,1)^{T}
\end{equation}
and
\begin{equation}\label{eq8}
G(w_1,w_2,\ldots,w_n)^{T}=\boldsymbol{e}_{h+1}.
\end{equation}
If we put
\[
H=
\begin{pmatrix}
u_1\alpha_1^{n-1} & \cdots & u_1\alpha_1^{n-h} & w_1 & u_1\alpha_1^{\,n-h-2} & \cdots & u_1\\
\vdots & \ddots & \vdots & \vdots & \vdots & \ddots & \vdots\\
u_n\alpha_n^{n-1} & \cdots & u_n\alpha_n^{n-h} & w_n & u_n\alpha_n^{\,n-h-2} & \cdots & u_n
\end{pmatrix}_{n\times n},
\]
then it follows from Eq.\eqref{eq7} and Eq.\eqref{eq8} and after a simple calculation that the $n\times n$ matrix $GH$  given by
$$
\small
GH  =
\begin{array}{@{}c@{\;}c@{}}
\left(
\begin{array}{c|c}
\begin{array}{ccc}
1 & \cdots & 0 \\
\vdots & \ddots & \vdots \\
* & \cdots & 1
\end{array}
&
\begin{array}{ccccc}
~~~~0 &~~~~~~~~~~~~~~ \cdots &&&~~~~~~\,~~0 \\
~~~~\vdots & && & ~~~~~~~~\,\vdots \\
~~~~0 & ~~~~~~~~~~~~~~\cdots&&&~~~~~~~\,~0
\end{array}
\\ \hline
\begin{array}{ccc}
* & \cdots & * \\
\vdots &  & \vdots \\
* & \cdots & *
\end{array}
&
\begin{array}{ccccc}
1 & 0 & 0 & \cdots& 0 \\
S_1 & 1 & 0 & \cdots& 0 \\
S_2 & S_1 & 1 & \cdots &0 \\
\vdots & \vdots & \vdots & \ddots& \vdots \\
S_{n-k-1} & S_{n-k-2} & S_{n-k-3} & \cdots& 1
\end{array}
\end{array}
\right)\;
\\[-0.2ex]
\multicolumn{1}{c}{%
\begin{array}{@{}c@{\hspace{\arrayrulewidth}}c@{}}
\hphantom{\begin{array}{ccc}
1 & \cdots & 0\\
\vdots & \ddots & \vdots\\
* & \cdots & 1
\end{array}}
&
\begin{array}{@{}ccccc@{}}
\scriptstyle ~k+1
& \phantom{S_{n-k-2}}
& \phantom{S_{n-k-3}}
& \phantom{\cdots}
& \phantom{1}
\end{array}
\end{array}%
}
&
\end{array}
\raisebox{0.99\baselineskip}{\(\scriptstyle k+1\)},
$$
\normalsize
is a lower triangular matrix over $\mathbb{F}_q$ whose main diagonal entries are
all equal to 1, where the entries denoted by $\ast$ may be nonzero and the bottom-right block matrix is an $(n-k)\times (n-k)$ submatrix.


Put $P=P_{k+l+1}P_{k+l}\cdots P_{k+1}$ and $P'=P'_{k+1}\cdots P'_{k+l+1}$, where $P_m=E_{h+1,m}(\eta_{m-k-1})$ and $P'_m=E_{h+1,m}(-\eta_{m-k-1})$ for $k+1\le m\le k+l+1$, and $E_{i,j}(\lambda)$ denotes the elementary matrix that adds $\lambda$ times the $j$-th row to the $i$-th row (resp. $\lambda$ times the $i$-th column to the $j$-th column) by left (resp. right) multiplication.
Then we have
$$PGHP'=\left(\begin{array}{cccc}
1 & 0  & \cdots& 0 \\
\ast & 1 &  \cdots& 0 \\
\vdots &  \vdots & \ddots& \vdots \\
\ast & \ast & \cdots& 1
\end{array}\right),$$
which is still a unit lower triangular matrix over $\mathbb{F}_q$. 

On the other hand, let $H'_h$ be the submatrix consisting of the last $(n-k)$ columns of $HP'$ and $G'_{h}$ the submatrix consisting of the first $k$ rows of $PG$. It is easy to verify that $$G'_{h}H_{h}^{'}=\boldsymbol{O}.$$
Specifically, ${H'_{h}}^{T}=\operatorname{diag}\{v_1,v_2,\ldots,v_n\}H_h$ and $\operatorname{diag}\{v_1,v_2,\ldots,v_n\}G'_{h}=G_{h}$, which implies that $G_{h}H_{h}=\boldsymbol{O}$. Moreover, since $\det(GH)\neq0$, we have $\det(H)\neq0$, which shows that $\operatorname{rank}(H_{h})=n-k$, and thus $H_{h}$ is the parity check matrix of $\mathcal{C}_{h}$. This completes the proof.
\end{proof}

\begin{remark}
If we set $h=0$, the result coincides with that of Theorem 3.1 in [1]. Moreover, if $\mathcal{L}=\{0\}$, then our result is consistent with Theorem 2.8 in [2].
\end{remark}


\section{A necessary and sufficient condition for $\mathcal{C}_h$ to be AMDS}
In this section, we are concerned with the AMDS property of the code $\mathcal{C}_h$, and we begin by the following lemma.
\begin{lemma}\label{fth}$[\ref{t37}]$
Let $\boldsymbol{\alpha}=(\alpha_1,\alpha_2,\ldots,\alpha_n)\in\mathbb{F}_q^n$ with $\alpha_1,\alpha_2,\ldots,\alpha_n$  pairwise distinct, and let $\mathcal{I}=\{i_1,i_2,\ldots,i_k\}$ be a $k$-subset of
$\{1,2,\ldots,n\}$. Let
$g(x)=\prod_{j=1}^k (x-\alpha_{i_j})=\sum_{j=0}^k c_j x^{k-j}$, where the coefficients $c_j$ are uniquely determined by the values of
$\alpha_{i_1},\ldots,\alpha_{i_k}$. For any $0\le t\le n-k-1$, define $f_{r,s}\in\mathbb{F}_q$ uniquely
for $0\le r\le k-1$ such that
\begin{equation*}
(\alpha_{i_1}^{k+t},\alpha_{i_2}^{k+t},\ldots,\alpha_{i_k}^{k+t})
=(f_{t,0},f_{t,1},\ldots,f_{t,k-1})
\begin{pmatrix}
1 & 1 & \cdots & 1\\
\alpha_{i_1} & \alpha_{i_2} & \cdots & \alpha_{i_k}\\
\vdots & \vdots & \ddots & \vdots\\
\alpha_{i_1}^{k-1} & \alpha_{i_2}^{k-1} & \cdots & \alpha_{i_k}^{k-1}
\end{pmatrix}.
\end{equation*}
Then we have
\begin{equation*}
f_{t,r}=-\sum_{i=0}^{t}c_{k-r+i}\,e_{t-i}=-(e_t,e_{t-1},\ldots,e_0)\boldsymbol{c}^{(r)}_t
=-\boldsymbol{e}_{t+1}^T A_{\mathcal{I},t}^{-1}\boldsymbol{c}^{(r)}_t.
\end{equation*}
where $\boldsymbol{e}_{t+1}=(0,\ldots,0,1)^{T}$,
$\boldsymbol{c}^{(r)}_t=(c_{k-r},c_{k-r+1},\ldots,c_{k-r+t})^{T}$ and
$$
A_{\mathcal{I},t}=\begin{pmatrix}
c_0 & 0   & \cdots & 0\\
c_1 & c_0 & \cdots & 0\\
\vdots & \vdots & \ddots & \vdots\\
c_t & c_{t-1} & \cdots & c_0
\end{pmatrix}.
$$
\end{lemma}

The following lemma is a special case of Theorem~1 in [1] and provides a necessary and sufficient condition for the code $C_h$ to be MDS, which will be used in the proof of the subsequent theorem. For completeness, we give a simple proof.

\begin{lemma}\label{ChMDS}$[\ref{t37}]$
Let the notion and symbols be defined
as those in Lemma~\ref{fth}. Then the $(\mathcal{L},\mathcal{P})$-TGRS code $\mathcal{C}_{h}$ defined by Eq.\eqref{Ch} is MDS if and only if $\boldsymbol{\eta}\in\Phi$, where $\Phi=\Big\{\boldsymbol{\eta}\in \mathbb{F}_q^{l+1}\setminus\{0\}\,|\,\forall\,\mathcal{I}=\{i_1,\cdots,i_k\}\subsetneq\{1,\ldots,n\}, 1-\sum_{t\in\mathcal{L}}\eta_t\,\boldsymbol{e}_{t+1}^T A_{\mathcal{I},t}^{-1}\boldsymbol{c}^{(h)}_t\neq0\Big\}$.
\end{lemma}

\begin{proof}
By Lemma \ref{MDS}, $\mathcal{C}_{h}$ is MDS if and only if the determinant of any $k\times k$ submatrix of $G_h$ is nonzero. For $\mathcal{I}=\{i_1,\cdots,i_k\}\subsetneq\{1,\ldots,n\}$, the submatrix $G_I$ of $G_h$ corresponding to $\mathcal{I}$ is given by
$$G_I=\left(
    \begin{array}{ccc}
      1 & \cdots & 1 \\
      \vdots & \ddots & \vdots \\
      \alpha_{i_1}^{h-1} & \cdots & \alpha_{i_k}^{h-1} \\
      \alpha_{i_1}^{h}+\sum_{t\in \mathcal{L}} \eta_{t}\,\alpha_{i_1}^{k+t} &\cdots & \alpha_{i_k}^{h}+\sum_{t\in \mathcal{L}} \eta_{t}\,\alpha_{i_k}^{k+t} \\
      \alpha_{i_1}^{h+1} & \cdots & \alpha_{i_k}^{h+1} \\
      \vdots & \ddots & \vdots \\
      \alpha_{i_1}^{k-1} & \cdots & \alpha_{i_k}^{k-1} \\
    \end{array}
  \right)_{k\times k}.
$$
Then by Lemma $\ref{fth}$ and linearity of the determinant in a row, we have
\begin{align*}
\det(G_I)&=\det\left(
    \begin{array}{ccc}
      1 & \cdots & 1 \\
      \alpha_{i_1} & \cdots & \alpha_{i_k} \\
      \vdots & \ddots & \vdots \\
      \alpha_{i_1}^{k-2} &\cdots & \alpha_{i_k}^{k-2} \\
      \alpha_{i_1}^{k-1} & \cdots & \alpha_{i_k}^{k-1} \\
    \end{array}
  \right)+\sum_{t\in \mathcal{L}} \eta_{t}\cdot \det
  \left(
    \begin{array}{ccc}
      1 & \cdots & 1 \\
      \vdots & \ddots & \vdots \\
      \alpha_{i_1}^{h-1} & \cdots & \alpha_{i_k}^{h-1} \\
      \alpha_{i_1}^{k+t} &\cdots & \alpha_{i_k}^{k+t} \\
      \alpha_{i_1}^{h+1} & \cdots & \alpha_{i_k}^{h+1} \\
      \vdots & \ddots & \vdots \\
      \alpha_{i_1}^{k-1} & \cdots & \alpha_{i_k}^{k-1} \\
    \end{array}
  \right)\\
&=\prod_{1\le m<j\le k}(\alpha_j-\alpha_m)+\sum_{t\in\mathcal L}\eta_t\,\det\begin{pmatrix}
1 & \cdots & 1\\
\alpha_{i_1} & \cdots & \alpha_{i_k}\\
\vdots & \ddots & \vdots\\
\alpha_{i_1}^{h-1} & \cdots & \alpha_{i_k}^{h-1}\\
\sum_{s=0}^{k-1} f_{t,s}\alpha_{i_1}^{s} & \cdots & \sum_{s=0}^{k-1} f_{t,s}\alpha_{i_k}^{s}\\
\alpha_{i_1}^{h+1} & \cdots & \alpha_{i_k}^{h+1}\\
\vdots & \ddots & \vdots\\
\alpha_{i_1}^{k-1} & \cdots & \alpha_{i_k}^{k-1}
\end{pmatrix}\\
&=\prod_{1\le m<j\le k}(\alpha_{i_j}-\alpha_{i_m})+\sum_{t\in\mathcal L}\eta_t\,f_{t,h}\prod_{1\le m<j\le k}(\alpha_{i_j}-\alpha_{i_m})\\
&=\left(1-\sum_{t\in\mathcal L}\eta_t\,\boldsymbol{e}_{t+1}^T A_{\mathcal{I},t}^{-1}\boldsymbol{c}^{(h)}_t\right)\prod_{1\le m<j\le k}(\alpha_{i_j}-\alpha_{i_m}).
\end{align*}
So it follows that $\det(G_I)\neq0$ if and only if $1-\sum_{t\in\mathcal L}\eta_t\,\boldsymbol{e}_{t+1}^T A_{I,t}^{-1}\boldsymbol{c}^{(h)}_t\neq0$ since $\prod_{1\le m<j\le k}(\alpha_{i_j}-\alpha_{i_m})\neq0$. This completes the proof.
\end{proof}

In the following, we present a necessary and sufficient condition for the
$(\mathcal{L},\mathcal{P})$-TGRS code $\mathcal{C}_{h}$ to be AMDS.

\begin{theorem}Let the notion and symbols be defined
as those in Lemma \ref{ChMDS}. The $(\mathcal{L},\mathcal{P})$-TGRS code $\mathcal{C}_{h}$ defined by Eq.\eqref{Ch} is AMDS if and only if the following two conditions hold simultaneously,
\begin{itemize}
  \item [$(1)$] $\boldsymbol{\eta}\notin\Phi$;
  \item [$(2)$] for any $(k+1)$-subset $J=\{j_1,j_2,\ldots,j_{k+1}\}\subsetneq \{1,2,\ldots,n\}$, there exists some $k$-subset $\mathcal{I}\subsetneq J$ such that
$$1-\sum_{t\in\mathcal{L}}\eta_t\,\boldsymbol{e}_{t+1}^T A_{\mathcal{I},t}^{-1}\boldsymbol{c}^{(h)}_t\neq0.$$
\end{itemize}
\end{theorem}
\begin{proof}
(1) By Lemma \ref{AMDS}, it suffices to show that there exist $k$ linearly dependent columns in $G_h$. By Lemma \ref{ChMDS}, for the $k\times k$ submatrix $G_\mathcal{I}$ of $G_h$ indexed by $\mathcal{I}$, $\det(G_\mathcal{I})\neq0$ if and only if $1-\sum_{t\in\mathcal L}\eta_t\,\boldsymbol{e}_{t+1}^T A_{\mathcal{I},t}^{-1}\boldsymbol{c}^{(h)}_t\neq0$.
Therefore, there exist $k$ linearly dependent columns in $G_h$ if and only if $\boldsymbol{\eta}\notin\Phi$.

(2) For the $(k+1)\times (k+1)$ submatrix $G_J$ consisted of any $k+1$ columns in $G_h$ given by
$$G_J=\left(
    \begin{array}{ccc}
      1 & \cdots & 1 \\
      \vdots & \ddots & \vdots \\
      \alpha_{i_1}^{h-1} & \cdots & \alpha_{i_{k+1}}^{h-1} \\
      \alpha_{i_1}^{h}+\sum_{t\in \mathcal{L}} \eta_{t}\,\alpha_{i_1}^{k+t} &\cdots & \alpha_{i_k}^{h}+\sum_{t\in \mathcal{L}} \eta_{t}\,\alpha_{i_{k+1}}^{k+t} \\
      \alpha_{i_1}^{h+1} & \cdots & \alpha_{i_{k+1}}^{h+1} \\
      \vdots & \ddots & \vdots \\
      \alpha_{i_1}^{k-1} & \cdots & \alpha_{i_{k+1}}^{k-1} \\
    \end{array}
  \right),$$
it is obvious that $\mathrm{rank}(M)\leq k$. Thus, $\mathrm{rank}(M)=k$ if and only if there exist some $k\times k$ nonsingular submatrix in $G_J$, that is, for any $(k+1)$-subset $J\subseteq \{1,2,\ldots,n\}$, there exist some $k$-subset $\mathcal{I}\subseteq J$ such that
$$1-\sum_{t\in\mathcal{L}}\eta_t\,\boldsymbol{e}_{t+1}^T A_{\mathcal{I},t}^{-1}\boldsymbol{c}^{(h)}_t\neq0.$$
This completes the proof.
\end{proof}

\section{LCD codes constructed from the code $\mathcal{C}_h$ }
In this section, we study two classes of LCD codes from the code $\mathcal{C}_h$ under special choices of evaluation points $\boldsymbol{\alpha}$ and $\boldsymbol{v}$, from which two classes of LCD MDS codes are further constructed. The following theorem presents the first class of LCD codes obtained from  $\mathcal{C}_h$.

\begin{theorem}\label{th1LCD}
Let $k$ and $n$ be two positive integers with $2\leq k \le \frac{n-2l-1}{2}$, $n \mid (q-1)$, and $\lambda \in \mathbb{F}_q^{\ast}$ with $\operatorname{ord}(\lambda)\mid \frac{q-1}{n}$. Let $\boldsymbol{\alpha}=(\alpha_1,\alpha_2,\ldots,\alpha_n)$ and $\boldsymbol{v}=(v_1,\ldots,v_n)$ satisfy that $\alpha_1,\ldots,\alpha_n$ be all roots of $m(x)=x^n-\lambda$, $v_i\in\{1,-1\}$ for $k\le i\le n$ and
$v_i\in\mathbb{F}_q\setminus\{-1,0,1\}$ for $1\le i\le k-1$, respectively. Set $r_{h-1}=\mathrm{Coeff}_{h-1}\bigl(q(x)\bmod g(x)\bigr)$, where $q(x)=x^{h-1}+\sum_{t\in\mathcal L}\eta_t x^{k+t-1}$ and $g(x)=\prod_{i=1}^{k-1}(x-\alpha_i)$.
If $r_{h-1}\neq 0$, then $\mathcal{C}_h$ is an LCD code.
\end{theorem}

\begin{proof}

With the above choices of $\boldsymbol{\alpha}$ and $\boldsymbol{v}$, the code $\mathcal{C}_h$ is a TGRS code. This is because $m'(x)=n x^{n-1}$ and $\gcd(m(x),m'(x))=1$ imply that the roots $\alpha_1,\alpha_2,\ldots,\alpha_n$ of $m(x)=x^n-\lambda$ are pairwise distinct in the splitting field over $\mathbb{F}_q$, and $\operatorname{ord}(\lambda)\mid \frac{q-1}{n}$ further yields $\boldsymbol{\alpha}\in\mathbb{F}_q^n$.

Let $G_h$ be the generator matrix of $\mathcal{C}_h$ given by Eq.\eqref{G}. Then according to Threorem \ref{th1}, $H_h$ in Eq.\eqref{H} is a parity check matrix of $\mathcal{C}_h$. By Lemma~\ref{LCD}, to prove that $\mathcal{C}_h$ is LCD, it suffices to show that $\operatorname{rank}\!\begin{pmatrix} G_h\\ H_h\end{pmatrix}=n$. To this end, we set $D=\operatorname{diag}\{v_1,v_2,\dots,v_{n}\}\begin{pmatrix} G_h\\ H_h\end{pmatrix}$ and prove that $\operatorname{rank}(D)=n$.

From the definitions of $u_i$ and $\alpha_i$ for $1 \le i \le n$, we have $u_i=\prod_{j\ne i}(\alpha_i-\alpha_j)^{-1}=\frac{1}{m'(\alpha_i)}=\frac{1}{n\alpha_i^{n-1}}=\frac{\alpha_i}{n\lambda}$ since $\alpha_i^n=\lambda$. Then by substituting this expression for $u_i$ into $D$, we obtain
\begingroup
\footnotesize
\setlength{\arraycolsep}{3pt}
\renewcommand{\arraystretch}{1.1}
\[
D=\left(
\begin{array}{ccc:ccc}
v_1^2 & \cdots & v_{k-1}^2 & 1 & \cdots & 1\\
\hdashline
v_1^2\alpha_1 & \cdots & v_{k-1}\alpha_{k-1}^2 & \alpha_k & \cdots & \alpha_k\\
\vdots & \ddots & \vdots & \vdots & \ddots & \vdots\\
v_1^2\alpha_1^{h-1} & \cdots & v_{k-1}^2\alpha_{k-1}^{h-1} & \alpha_k^{h-1} & \cdots & \alpha_k^{h-1}\\
v_1^2\!\big(\alpha_1^{h}\!+\!\sum\limits_{t\in\mathcal L}\eta_t\alpha_1^{k+t}\big) & \cdots &
v_{k-1}^2\!\big(\alpha_{k-1}^{h}\!+\!\sum\limits_{t\in\mathcal L}\eta_t\alpha_{k-1}^{k+t}\big) &
\alpha_k^{h}\!+\!\sum\limits_{t\in\mathcal L}\eta_t\alpha_k^{k+t} & \cdots &
\alpha_n^{h}\!+\!\sum\limits_{t\in\mathcal L}\eta_t\alpha_n^{k+t}\\[0.5em]
v_1^2\alpha_1^{h+1} & \cdots & v_{k-1}^2\alpha_{k-1}^{h+1} & \alpha_k^{h+1} & \cdots & \alpha_n^{h+1}\\
\vdots & \ddots & \vdots & \vdots & \ddots & \vdots\\
v_1^2\alpha_1^{k-1} & \cdots & v_{k-1}^2\alpha_{k-1}^{k-1} & \alpha_k^{k-1} & \cdots & \alpha_n^{k-1}\\
\hdashline
\frac{1}{n\lambda}\alpha_1 & \cdots & \frac{1}{n\lambda}\alpha_{k-1} & \frac{1}{n\lambda}\alpha_k & \cdots & \frac{1}{n\lambda}\alpha_n\\
\vdots & \ddots & \vdots & \vdots & \ddots & \vdots\\
\frac{1}{n\lambda}\alpha_1^{\,n-k-(l+1)} & \cdots &
\frac{1}{n\lambda}\alpha_{k-1}^{\,n-k-(l+1)} &
\frac{1}{n\lambda}\alpha_k^{\,n-k-(l+1)} & \cdots &
\frac{1}{n\lambda}\alpha_n^{\,n-k-(l+1)}\\[2mm]
\frac{\big(\alpha_1^{\,n-k-l}-\widetilde w_1\alpha_1\Psi_l\big)}{n\lambda} & \cdots &
\frac{\big(\alpha_{k-1}^{\,n-k-l}-\widetilde w_{k-1}\alpha_{k-1}\Psi_l\big)}{n\lambda} &
\frac{\big(\alpha_k^{\,n-k-l}-\widetilde w_k\alpha_k\Psi_l\big)}{n\lambda} & \cdots &
\frac{\big(\alpha_n^{\,n-k-l}-\widetilde w_n\alpha_n\Psi_l\big)}{n\lambda}\\
\vdots & \ddots & \vdots & \vdots & \ddots & \vdots\\
\frac{\big(\alpha_1^{\,n-k}-\widetilde w_1\alpha_1\Psi_0\big)}{n\lambda} & \cdots &
\frac{\big(\alpha_{k-1}^{\,n-k}-\widetilde w_{k-1}\alpha_{k-1}\Psi_0\big)}{n\lambda} &
\frac{\big(\alpha_k^{\,n-k}-\widetilde w_k\alpha_k\Psi_0\big)}{n\lambda} & \cdots &
\frac{\big(\alpha_n^{\,n-k}-\widetilde w_n\alpha_n\Psi_0\big)}{n\lambda}
\end{array}
\right),
\]
\endgroup
where $\Psi_s:=\sum_{t=s}^{l}\eta_t\,S_{t-s}$ for $0\le s\le l$.

If $2\leq k \le \frac{n-2l-1}{2}$, i.e., $k+l\leq n-k-(l+1)$, then by applying elementary row operations to the matrix $D$, we obtain the following matrix
$$D_1=\begin{pmatrix}
A & \boldsymbol{O}\\
B & C
\end{pmatrix},$$
where
$$\small
A=\begin{pmatrix}
(v_1^2-1)\alpha_1 & \cdots & (v_{k-1}^2-1)\alpha_{k-1} \\
\vdots & \ddots & \vdots \\
(v_1^2-1)\big(\alpha_1^{h}+\sum\limits_{t\in\mathcal L}\eta_t\alpha_1^{k+t}\big) & \cdots &
(v_{k-1}^2-1)\big(\alpha_{k-1}^{h}+\sum\limits_{t\in\mathcal L}\eta_t\alpha_{k-1}^{k+t}\big) \\
\vdots & \ddots & \vdots  \\
(v_1^2-1)\alpha_1^{k-1} & \cdots & (v_{k-1}^2-1)\alpha_{k-1}^{k-1} \\
\end{pmatrix}_{(k-1)\times (k-1)},$$
and
$$\small C=
\begin{pmatrix}
1  & \cdots & 1 \\
\dfrac{1}{n\lambda} \alpha_k  & \cdots & \dfrac{1}{n\lambda} \alpha_n\\
\vdots  & \ddots & \vdots\\
\dfrac{1}{n\lambda} \alpha_k^{n-k-(l+1)}  &
\cdots &
\dfrac{1}{n\lambda} \alpha_n^{n-k-(l+1)}\\[2mm]
\dfrac{1}{n\lambda} \!\Big(\alpha_k^{n-k-l}-\widetilde{w}_k\alpha_k\Psi_l\Big)  &
\cdots &
\dfrac{1}{n\lambda}\!\Big(\alpha_n^{n-k-l}- \widetilde{w}_n\alpha_n\Psi_l\Big)\\
\vdots  & \ddots & \vdots\\
\dfrac{1}{n\lambda}\!\Big(\alpha_k^{n-k}-\widetilde{w}_k\alpha_k\Psi_0\Big) &
\cdots &
\dfrac{1}{n\lambda}\!\Big(\alpha_n^{n-k}-\widetilde{w}_n\alpha_n\Psi_0\Big)
\end{pmatrix}_{(n-k+1)\times (n-k+1)}.$$
Therefore, to show that $D_1$ is invertible, it remains to prove that $\operatorname{rank}(A)=k-1$ and $\operatorname{rank}(C)=n-k+1$. Next, we divide the proof into the following two steps.

\textbf{Step 1.} We prove that $\operatorname{rank}(A)=k-1$.

Note that the matrix $A$ can be decomposed as $A=D_v\cdot D_\alpha\cdot \widetilde{A}$, where $D_v=\operatorname{diag}\{v_1^2-1,\dots,v_{k-1}^2-1\}$, $D_\alpha=\operatorname{diag}\{\alpha_1,\dots,\alpha_{k-1}\}$ and $\widetilde{A}$ is the Vandermonde matrix $(\alpha_j^{i-1})_{1\le i,j\le {k-1}}$ with its $h$-th row being replaced by
$(q(\alpha_1),\dots,q(\alpha_{k-1}))$. It is clear that $\det(D_v)\neq0$ and $\det(D_\alpha)\neq0$ since $(v_i^2-1)\alpha_i\neq 0$ for $1\le i\le k-1$.

Now, we examine the value of $\det(\widetilde{A})$. From the definitions of $g(x)$ and $q(x)$, performing the polynomial division of $q(x)$ by $g(x)$ yields
\begin{align}\label{q(x)}
q(x)=b(x)g(x)+r(x),\quad \deg(r(x))\le k-2, ~~r(x)=\sum_{i=0}^{k-2}r_{i}x^{i},
\end{align}
which leads to $q(\alpha_j)=r(\alpha_j)=\sum_{i=0}^{k-2}r_i\,\alpha_j^{\,i}$ for $1\le j\le k-1$ since $g(\alpha_j)=0$. Hence,
\[
\det(\widetilde{A})=r_{h-1}\prod_{1\le i<j\le k-1}(\alpha_j-\alpha_i).
\]
If $r_{h-1}\neq 0$, then
\[
\det(A)=\det(D_v)\det(D_\alpha)\det(\widetilde{A})\neq 0,
\]
that is, $\operatorname{rank}(A)=k-1$.

\textbf{Step 2.} We show that $\operatorname{rank}(C)=n-k+1$.

Since $\alpha_1,\alpha_2,\ldots,\alpha_n$ are pairwise distinct roots of $x^n-\lambda$ in $\mathbb{F}_q$, we have
\begin{align}\label{eq1}
\prod_{j\ne i}(x-\alpha_j)=\frac{x^n-\lambda}{x-\alpha_i}=\frac{x^n-\alpha_i^n}{x-\alpha_i}
=x^{n-1}+\alpha_i x^{n-2}+\cdots+\alpha_i^{n-1}.
\end{align}
Meanwhile, we have
\[
\prod_{j\ne i}(x-\alpha_j)=\sum_{t=0}^{n-1}(-1)^{n-1-t}\sigma_{n-1-t}(i)x^t.
\]
Comparing the coefficients $x^r$ in Eq.\eqref{eq1} we obtain $\sigma_r(i)=(-1)^r\alpha_i^r$ for all $0\le r\le n-1$. Consequently, $\widetilde w_i=(-1)^{n+h+1}\sigma_{n-1-h}(i)=\alpha_i^{n-1-h}$, that is, $\alpha_i\widetilde w_i=\alpha_i^{n-h}$.

Therefore, substituting $\alpha_i\widetilde w_i=\alpha_i^{n-h}$ into the matrix $C$,we obtain the following matrix decomposition
\[
C=\widehat{C}\cdot \operatorname{diag}\{1,\underbrace{\tfrac1{n\lambda},\ldots,\tfrac1{n\lambda}}_{n-k\text{ times}}\},
\]
where
\[
\small
\widehat{C}=
\begin{pmatrix}
1  & \cdots & 1 \\
\alpha_k  & \cdots &  \alpha_n\\
\vdots  & \ddots & \vdots\\
\alpha_k^{n-k-l-1}  &\cdots &\alpha_n^{n-k-l-1}\\
\alpha_k^{n-k-l}-\alpha_k^{n-h}\Psi_l & \cdots & \alpha_n^{n-k-l}-\alpha_n^{n-h}\Psi_l\\
\vdots  & \ddots & \vdots\\
\alpha_k^{n-k}-\alpha_k^{n-h}\Psi_0 & \cdots & \alpha_n^{n-k}-\alpha_n^{n-h}\Psi_0
\end{pmatrix}.
\]

Then, we present a matrix decomposition of $\widehat{C}$. If we put $f(x)=\prod_{i=k}^{n}(x-\alpha_i)$, i.e., $f(x)=\frac{m'(x)}{g(x)}$. By the Euclidean division algorithm, there exist unique polynomials
$t(x),\rho(x)\in \mathbb{F}_q[x]$ such that
\begin{align}\label{x_n-h}
x^{n-h}=t(x)f(x)+\rho(x),
\end{align}
where $\rho(x)=\sum_{j=0}^{n-k}\rho_jx^j$ with $\rho_j\in \mathbb{F}_q$.
Evaluating at $x=\alpha_i$ in Eq.\eqref{x_n-h} yields
\[
\alpha_i^{n-h}=\rho(\alpha_i)=\sum_{j=0}^{n-k}\rho_j\alpha_i^j \qquad (k\le i\le n),
\]
that is, $(\alpha_k^{n-h},\alpha_{k+1}^{n-h},\ldots,\alpha_{n}^{n-h})=\boldsymbol{\rho}L$, where $\boldsymbol{\rho}=(\rho_0,\ldots,\rho_{n-k})$ and  $L=\bigl(\alpha_{k+j-1}^{\,i-1}\bigr)_{1\le i,j\le n-k+1}$ be the $(n-k+1)\times (n-k+1)$ Vandermonde matrix. Clearly, $\det(L)\neq 0$ since $\alpha_k,\alpha_{k+1},\ldots,\alpha_{n}$ are pairwise distinct.
Then we have the following matrix factorization
$$\widehat{C}=ML,$$
where $M=I_{n-k+1}-\boldsymbol{u}^{T}\boldsymbol{\rho}$ and $\boldsymbol{u}=(\underbrace{0,\ldots,0}_{n-k-l},\Psi_l,\Psi_{l-1},\ldots,\Psi_0)$.

By Lemma~\ref{l6}, we have
\[
\det(M)=\det\bigl(I_{n-k+1}-\boldsymbol{u}^{\mathsf T}\boldsymbol{\rho}\bigr)
=1-\boldsymbol{\rho}\boldsymbol{u}^{\mathsf T}
=1-\sum_{t=0}^{l}\Psi_t\, \rho_{\,n-k-t}.
\]

Now, we claim that  $r_{h-1}=1-\sum_{t=0}^{l}\Psi_t\, \rho_{\,n-k-t}$, and we prove this claim below.

Since $\alpha_1,\alpha_2,\ldots,\alpha_n$ are pairwise distinct roots of $x^n-\lambda$, we have
\[
\prod_{i=1}^n (1-\alpha_i z)=1-\lambda z^n,
\]
So, it follows from Definition \ref{matrix} that
$$
\sum_{t\ge 0} S_t(\alpha_1,\ldots,\alpha_n)\,z^t=\prod_{i=1}^n \frac{1}{1-\alpha_i z}=\frac{1}{1-\lambda z^n}
=\sum_{j\ge 0}\lambda^j z^{nj},
$$
which implies $S_0=1$, $S_t=0$ for $1\le t< n$ and $S_{nj}=\lambda^j$. Note that $l\leq n-k-1<n$. Hence,  for $0\le s\le l$, we get \begin{align}\label{eq9}
\Psi_s=\sum_{t=s}^{l}\eta_t\,S_{t-s}=\eta_sS_0=\eta_s, \end{align}
and thus
\begin{align}\label{eq3}
1-\sum_{t=0}^l \eta_t\,s_{n-k-t}=1-\sum_{t=0}^l \Psi_t\,s_{n-k-t}.
\end{align}

Furthermore, multiplying both sides of Eq.\eqref{q(x)} by $(x^{n-h}-\rho(x))$, we obtain
\begin{align}\label{eq2}
b(x)g(x)\bigl(x^{n-h}-\rho(x)\bigr)+r(x)\bigl(x^{n-h}-\rho(x)\bigr)=q(x)(x^{n-h}-\rho(x)\bigr).
\end{align}
Then by Eq.\eqref{x_n-h}, we have $b(x)g(x)(x^{n-h}-\rho(x))=(x^n-\lambda)b(x)t(x)$ since $m(x)=g(x)f(x)=x^n-\lambda$, where $\deg (b(x))\leq(k+l-1)-(k-1)=l$ and $\deg (t(x))\le (n-h)-(n-k+1)=k-1-h$. Recall that $h\geq1$ and $l\leq n-k-1$, then it follows from $\deg(b(x)t(x))\le k-1-h+l\le n-3$ that
\[
\mathrm{Coeff}_{n-1}\bigl(b(x)g(x)(x^{n-h}-\rho(x))\bigr)=0.
\]
Consequently,
\begin{align}\label{eq4}
\mathrm{Coeff}_{n-1}\bigl(r(x)(x^{n-h}-\rho(x))\bigr)=\mathrm{Coeff}_{n-1}\bigl(q(x)(x^{n-h}-\rho(x))\bigr)
\end{align}
holds in Eq.\eqref{eq2}. It is clear that the left-hand side of Eq.\eqref{eq4} satisfies
\begin{align}\label{eq5}
\mathrm{Coeff}_{n-1}\bigl(r(x)(x^{n-h}-\rho(x))\bigr)=\mathrm{Coeff}_{n-1}\bigl(r(x)x^{n-h}\bigr)
=\mathrm{Coeff}_{h-1}(r(x))=r_{h-1}
\end{align}
since $\deg(r(x)\rho(x))\le (k-2)+(n-k)=n-2$. Meanwhile, substituting $q(x)=x^{h-1}+\sum_{t=0}^l \eta_t x^{k+t-1}$ into the right-hand side of Eq.\eqref{eq4} gives
\begin{align}\label{eq6}
\nonumber&\mathrm{Coeff}_{n-1}\Bigl(\Bigl(x^{h-1}+\sum_{t=0}^{l}\eta_t x^{k+t-1}\Bigr)\bigl(x^{n-h}-\rho(x)\bigr)\Bigr) \\
=&\mathrm{Coeff}_{n-1}\bigl(x^{h-1}x^{n-h}\bigr)-\sum_{t=0}^{l}\eta_t \mathrm{Coeff}_{n-1}\bigl(x^{k+t-1}\rho(x)\bigr) \\ \nonumber
=&1-\sum_{t=0}^{l}\eta_t\,\rho_{\,n-k-t},\nonumber
\end{align}
where the first equality uses the facts that $k+t-1+n-h\ge n+1$ for each $0\le t\le l$ and $\deg\bigl(x^{h-1}\rho(x)\bigr)\le (h-1)+(n-k)\le (k-2)+(n-k)=n-2$, and the second equality follows from $\rho(x)=\sum_{j=0}^{n-k}\rho_j x^j$.

Therefore, from Eqs.\eqref{eq3}, \eqref{eq4}-\eqref{eq6}, we have
\[
r_{h-1}=1-\sum_{t=0}^l \eta_t\,\rho_{n-k-t}=1-\sum_{t=0}^l \Psi_t\,\rho_{n-k-t},
\]
which implies that $1-\sum_{t=0}^l \Psi_t\,\rho_{n-k-t}\neq 0$ since $r_{h-1}\neq 0$.
Hence,
\[
\det(\widehat{C})=\det(M)\det(L)\neq 0.
\]
Therefore, $\operatorname{rank}(C)=n-k+1$.

In conclusion, $\operatorname{rank}(D)=\operatorname{rank}(D_1)=n$, and thus
$\mathcal{C}_h$ is an LCD code. This completes the proof.
\end{proof}

Combining Lemma~\ref{ChMDS} with Theorem~\ref{th1LCD}, we immediately obtain the following sufficient condition for the code $\mathcal{C}_h$ to be an LCD MDS code.

\begin{corollary}\label{coroLCDMDS}
Let the notion and symbols be defined
as in Lemma~\ref{ChMDS} and Theorem~\ref{th1LCD}. If $r_{h-1}\neq 0$ and $\boldsymbol{\eta}\in\Phi$,
then $\mathcal{C}_h$ is an LCD MDS code.
\end{corollary}

The following two examples obtained from  MAGMA program are presented to illustrate the validity of Theorem \ref{th1LCD} and Corollary \ref{coroLCDMDS}.

\begin{example}
Let $(q,n,k,h,l,\lambda)=(31,15,4,1,3,1)$ and $\mathcal{L}=\{0,1,2,3\}$.
Let $\boldsymbol{\alpha}=(\alpha_1,\alpha_2,\ldots,\alpha_{15})
=(2,20,25,1,4,5,7,8,9,10,14,16,18,19,28)$, where $\alpha_1,\ldots,\alpha_{15}$ be all roots of $x^{15}-1$. Choose $\boldsymbol v=(18,23,5,1,1,1,-1,1,-1,-1,1,-1,1,-1,1)$ and $\boldsymbol{\eta}=(\eta_0,\eta_1,\eta_2,\eta_3)=(5,21,12,14)$. A computation gives
\[
r_{h-1}=\mathrm{Coeff}_{h-1}\bigl(q(x)\bmod g(x)\bigr)=14\neq 0,
\]
Thus the code $\mathcal{C}_h$ generated
\[
G_h=\left(\begin{array}{ccccccccccccccc}
18&23&5&1&1&1&-1&1&-1&-1&1&-1&1&-1&1\\
8&10&16&22&27&17&-2&18&-2&1&20&21&12&17&27\\
10&24&25&1&16&25&13&2&12&24&10&23&14&11&9\\
20&15&5&1&2&1&-2&16&15&23&16&27&4&23&4
    \end{array}
  \right)
\]
is a $[15,4]$ LCD code. This result is consistent with the result of Theorem~\ref{th1LCD}.
\end{example}

\begin{example}
Let $(q,n,k,h,l,\lambda)=(37,9,3,1,1,1)$ and $\mathcal{L}=\{0,1\}$.
Let $\boldsymbol{\alpha}=(\alpha_1,\alpha_2,\ldots,\alpha_9)=(1,16,26,12,33,10,34,7,9)$, where $\alpha_1,\ldots,\alpha_{9}$ be all roots of $x^{9}-1$ in $\mathbb{F}_{37}$. Choose $\boldsymbol{v}=(21,30,1,1,-1,1,1,1,-1)$ and $\boldsymbol{\eta}=(\eta_0,\eta_1)=(22,24)$. Then
\[
r_{h-1}=\mathrm{Coeff}_{h-1}\bigl(q(x)\bmod g(x)\bigr)=3\neq 0.
\]
Thus the code $\mathcal{C}_h$ generated by
\[
G_h=\left(\begin{array}{ccccccccc}
21&30&1&1&-1&1&1&1&-1\\
25&33&6&6&4&13&15&20&19\\
21&21&10&33&21&26&9&12&30
\end{array}\right)
\]
is a $[9,3,7]$ LCD MDS code. This result is consistent with the result of Corollary \ref{coroLCDMDS}.
\end{example}

Now,we provide the second class of LCD codes obtained from  $\mathcal{C}_h$.

\begin{theorem}\label{th2LCD}
Let the notion and symbols be defined as in Theorem~\ref{th1LCD}, while $k = \frac{n-l-m}{2}\geq 2\,(0\leq m\leq l)$. If $r_{h-1}\neq 0$ and $\sum\limits_{t=m}^{l}\eta_t\eta_{l+m-t}=0$, then $\mathcal{C}_h$ is an LCD code.
\end{theorem}
\begin{proof}
It is evident that $\mathcal{C}_h$ is a TGRS code. Similar to Theorem \ref{th1LCD}, set $$D=\operatorname{diag}\{v_1,v_2,\ldots,v_n\}\begin{pmatrix} G_h\\ H_h\end{pmatrix},$$
where $D$, $G_h$ and $H_h$ have the same form as in Theorem \ref{th1LCD}. To prove that $\mathcal{C}_h$ is a LCD, it suffices to show that  $\operatorname{rank}(D)=n$.

Note that $n=2k+l+m\,(0\leq m\leq l)$, i.e., $n-k-l=k+m$. By applying elementary row operations to the matrix $D$ in a manner similar to that used in the proof of Theorem \ref{th1LCD}, we obtain an equivalent
matrix
$$D_2=\begin{pmatrix}
\widetilde{A} & \boldsymbol{O}\\
\widetilde{B} & \widetilde{C}
\end{pmatrix},$$
where
$$\small
\widetilde{A}=\begin{pmatrix}
(v_1^2-1)\alpha_1 & \cdots & (v_{k-1}^2-1)\alpha_{k-1} \\
\vdots & \ddots & \vdots \\
(v_1^2-1)\alpha_1^{h-1} & \cdots & (v_{k-1}^2-1)\alpha_{k-1}^{h-1} \\
(v_1^2-1)\alpha_1^{h+1} & \cdots & (v_{k-1}^2-1)\alpha_{k-1}^{h+1} \\
\vdots & \ddots & \vdots  \\
(v_1^2-1)\alpha_1^{k-1} & \cdots & (v_{k-1}^2-1)\alpha_{k-1}^{k-1} \\
\end{pmatrix},~~~~
\widetilde{B}
=\begin{pmatrix}
\boldsymbol{b}\\
B'
\end{pmatrix}
~~\text{and}~~
\widetilde{C}
=\begin{pmatrix}
\boldsymbol{c}\\
C'
\end{pmatrix},
$$
in which
\begin{align*}
&\boldsymbol{b}=\big((v_1^2\!-\!1)\big(\alpha_1^{h}\!+\!\sum\limits_{t\in\mathcal L}\eta_t\alpha_1^{k+t}\big)\!+\!\alpha_1^{n-h}\sum\limits_{t=m}^{l}\eta_t\Psi_{l+m-t},\ldots,\\ &~~~~~~~~~(v_{k-1}^2\!-\!1)\big(\alpha_{k-1}^{h}\!+\!\sum\limits_{t\in\mathcal L}\eta_t\alpha_{k-1}^{k+t}\big)\!+\!\alpha_{k-1}^{n-h}\sum\limits_{t=m}^{l}\eta_t\Psi_{l+m-t}\big)\in \mathbb{F}_q^{k-1}
\end{align*}
and $$\boldsymbol{c}=\big(\alpha_k^{n-h}\sum\limits_{t=m}^{l}\eta_t\Psi_{l+m-t},\ldots ,\\ \alpha_n^{n-h}\sum\limits_{t=m}^{l}\eta_t\Psi_{l+m-t}\big)\in\mathbb{F}_q^{n-k+1},$$
and the submatrices $B'$ and $C'$ have the same form as those of $B$ and $C$ in Theorem~\ref{th1LCD}, respectively.

On the other hand, since the definition of $\Psi_s$ is independent of $k$, the proof of Theorem~\ref{th1LCD} gives $\Psi_s=\eta_s$ for $0\le s\le l$ in Eq.\eqref{eq9}, and hence
$$\sum\limits_{t=m}^{l}\eta_t\Psi_{l+m-t}=\sum\limits_{t=m}^{l}\eta_t\eta_{l+m-t}=0$$ by assumption. Then $$\boldsymbol{b}=\big((v_1^2\!-\!1)\big(\alpha_1^{h}\!+\!\sum\limits_{t\in\mathcal L}\eta_t\alpha_1^{k+t}\big),\ldots, (v_{k-1}^2\!-\!1)\big(\alpha_{k-1}^{h}\!+\!\sum\limits_{t\in\mathcal L}\eta_t\alpha_{k-1}^{k+t}\big)\big)$$ and $\boldsymbol{c}=\boldsymbol{0}$, which implies
\[
D_2=
\begin{pmatrix}
\widetilde{A} & \boldsymbol{O}\\
\boldsymbol{b} & \boldsymbol{0}\\
B & C
\end{pmatrix}=\begin{pmatrix}
\widehat{A} & \boldsymbol{O}\\
B & C
\end{pmatrix},
\]
where $\widehat{A}=\begin{pmatrix}\widetilde{A} \\
\boldsymbol{b}\end{pmatrix}$ and $\widehat{A}$ has the same form as that of the matrix $A$ in Theorem~\ref{th1LCD} after a simple row permutation.

By an argument similar to the proof of Theorem~\ref{th1LCD} and the details are omitted for brevity, we obtain $\operatorname{rank}(\widehat{A})=k-1$ and $\operatorname{rank}(C)=n-k+1$. Then $\operatorname{rank}(D)=\operatorname{rank}(D_2)=n$. Therefore, $\mathcal{C}_h$ is an LCD code by Lemma \ref{LCD}. This completes the proof.
\end{proof}

Similarly, by Lemma~\ref{ChMDS} and Theorem~\ref{th2LCD}, we obtain the following corollary.

\begin{corollary}\label{2LCDMDS}
Let the notion and symbols be defined
as those in Lemma~$\ref{ChMDS}$ and Theorem~$\ref{th2LCD}$. If $r_{h-1}\neq 0$, $\sum\limits_{t=m}^{l}\eta_t\eta_{l+m-t}=0$ and $\boldsymbol{\eta}\in\Phi$, then $\mathcal{C}_h$ is an LCD MDS code.
\end{corollary}

The following two examples obtained from MAGMA program are presented to illustrate the validity of  Theorem \ref{th2LCD} and Corollary \ref{2LCDMDS}.

\begin{example}
Let $(q,n,l,m,h,\lambda)=(31,15,3,0,1,1)$ and $\mathcal{L}=\{0,1,2,3\}$. Then $k=\frac{n-l-m}{2}=6$. Let $\boldsymbol\alpha=(\alpha_1,\ldots,\alpha_{15})
=(1,5,8,25,28,2,4,7,9,10,14,16,18,19,20)$, where  $\alpha_1,\ldots,\alpha_{15}$ be all roots of $x^{15}-1$. Choose $\boldsymbol v=(25,21,22,23,6,1,1,1,1,-1,1,\\-1,1,-1,1)$ and $(\eta_0,\eta_1,\eta_2,\eta_3)=(3,21,22,1)$. Then
\[
r_{h-1}=\mathrm{Coeff}_{h-1}\bigl(q(x)\bmod g(x)\bigr)=15\neq 0~~\text{and}~~\sum_{t=0}^{3}\eta_t\Psi_{3-t}=0,
\]

and the code $\mathcal{C}_h$ generated
\[
G_h=\left(
      \begin{array}{ccccccccccccccc}
25&21&22&23&6&1&1&1&1&-1&1&-1&1&-1&1\\
22&25&5&20&5&5&1&-2&26&15&0&19&8&17&-1\\
25&-2&13&22&23&4&16&18&19&24&10&23&14&11&28\\
25&21&11&23&24&8&2&2&16&23&16&27&4&23&2\\
25&12&26&17&21&16&8&14&20&13&7&-2&10&3&9\\
25&-2&22&22&-1&1&1&5&25&6&5&-1&25&26&25
      \end{array}
    \right)
\]
is a $[15,6]$ LCD code. This result is consistent with the result of Theorem \ref{th2LCD}.
\end{example}

\begin{example}
Let $(q,n,l,m,h,\lambda)=(31,10,2,2,1,1)$ and $\mathcal{L}=\{0,1,2\}$. Then $k=\frac{n-l-m}{2}=3$. Let $\boldsymbol{\alpha}=(\alpha_1,\ldots,\alpha_{10})=(30,2,29,27,1,8,16,4,23,15)$, where $\alpha_1,\ldots,\alpha_{10}$ be all roots of $x^{10}-1$. Choose $\boldsymbol{v}=(22,15,-1,1,1,1,1,-1,-1,-1)$ and $\boldsymbol{\eta}=(\eta_0,\eta_1,\eta_2)=(28,6,0)\neq \boldsymbol{0}$. Then $\sum_{t=m}^{l}\eta_t\Psi_{l+m-t}=\eta_2\Psi_{2}=0$,
\[
q(x)=x^{h-1}+\sum_{t\in\mathcal L}\eta_t x^{k+t-1}=1+28x^{2}+6x^{3}
~\mathrm{and}~
g(x)=\prod_{i=1}^{k-1}(x-\alpha_i)=(x-30)(x-2),
\]
which leads to $r_{h-1}=\mathrm{Coeff}_{h-1}\bigl(q(x)\bmod g(x)\bigr)=7\neq 0$, and hence the code $\mathcal{C}_h$ generated by
\[
G_h=\left(\begin{array}{cccccccccc}
22&15&-1& 1& 1& 1& 1&-1&-1&-1\\
21&25& 6&19& 4&15&16&16&29&23\\
22&29&27&16& 1& 2& 8&15&29&23
\end{array}\right)
\]
is a $[10,3,8]$ LCD MDS code. This result is consistent with the result of Corollary \ref{2LCDMDS}.
\end{example}

%

\section{Conclusion}

In this paper, we have determined the parity check matrix of $(\mathcal{L},\mathcal{P})$-TGRS code $\mathcal{C}_h$ with $\mathcal{L}=\{0,1,\ldots,l\}$ for $l\leq n-k-1$ and $\mathcal{P}=\{h\}$ for $1\leq h\leq k-1$, and established a necessary and sufficient condition for $\mathcal{C}_h$ to be AMDS. Moreover, we have constructed two classes of LCD codes from $\mathcal{C}_h$ by choosing suitable evaluation points. In particular, our constructions have included LCD MDS codes.

We should mention that although our $(\mathcal{L},\mathcal{P})$-TGRS code can be viewed as a special case of the codes presented by Hu et al. in [\ref{t37}], the LCD and LCD MDS properties have not been investigated therein. Therefore, our work has provided a more refined characterization from this perspective. Moreover, compared with the constructions of Liang et al. in [\ref{t41},\ref{t47}], our approach has allowed more flexible selections of hooks, and our proof method is different from that adopted in their work.

\subsection*{Funding}

This work was partially supported by National Natural Science Foundation of China (Grant No.62172337; No.62562055), Key Project of Gansu Natural Science Foundation (Grant No. 23JRRA685), the Funds for Innovative Fundamental Research Group Project of Gansu Province (Grant No. 23JRRA684).


\begin{thebibliography}{99}
\bibitem{}\label{t42}
Huffman W C, Pless V. Fundamentals of Error-Correcting Codes. Cambridge University Press, Cambridge, 2003.

\bibitem{}\label{t1}
Manasse M, Thekkath C, Silverberg A. A Reed-Solomon code for disk storage, and efficient recovery computations for erasure-coded disk storage. Proceedings in Informatics, : 1-11, 2009.

\bibitem{}\label{t2}
Dau S H, Song W. Balanced sparsest generator matrices for MDS codes. 2013 IEEE International Symposium on Information Theory, Istanbul, Turkey, : 1889-1893, 2013.

\bibitem{}\label{t3}
Dinh H Q, Nguyen B T, Thi H L. AMDS constacyclic codes and quantum AMDS codes. Filomat, 38(33): 11889-11912, 2024.

\bibitem{}\label{t4}
Geng X, Yang M, Zhang J, Zhou Z. A class of almost MDS codes. Finite Fields and Their Applications, 79: 1-10, 2022.

\bibitem{}\label{t5}
Tang C, Ding C. An infinite family of linear codes supporting 4-designs. IEEE Transactions on Information Theory, 67(1): 244-254, 2021.

\bibitem{}\label{t6}
Xu G, Cao X, Qu L. Infinite families of 3-designs and 2-designs from almost MDS codes. IEEE Transactions on Information Theory, 68(7): 4344-4353, 2022.

\bibitem{}\label{t7}
Zhou Y, Wang F, Xin Y, Luo S, Qing S, Yang Y. A secret sharing scheme based on near-MDS codes. Proceedings of IC-NIDC 2009, : 833-836, 2009.

\bibitem{}\label{t8}
Massey J L. Linear codes with complementary duals. Discrete Mathematics, 106-107: 337-342, 1992.

\bibitem{}\label{t9}
Bringer J, Carlet C, Chabanne H, Guilley S, Maghrebi H. Orthogonal direct sum masking: A smartcard friendly computation paradigm in a code, with builtin protection against side-channel and fault attacks. Information Security Theory and Practice: Securing the Internet of Things (Lecture Notes in Computer Science), 8501: 40-56, 2014.

\bibitem{}\label{t10}
Carlet C, Guilley S. Complementary dual codes for countermeasures to side-channel attacks. Advances in Mathematics of Communications, 10(1): 131-150, 2016.



\bibitem{}\label{t16}
Beelen P, Puchinger S, Rosenkilde Nielsen J. Twisted Reed-Solomon codes. 2017 IEEE International Symposium on Information Theory, : 336-340, 2017.

\bibitem{}\label{t18}
Huang D, Yue Q, Niu Y. MDS or NMDS LCD codes from twisted Reed-Solomon codes. Cryptography and Communications, 15: 221-237, 2023.

\bibitem{}\label{t19}
Wu Y. Twisted Reed-Solomon codes with one-dimensional hull. IEEE Communications Letters, 25: 383-386, 2021.

\bibitem{}\label{t20}
Zhu C, Liao Q. The $[1, 0]$-twisted generalized Reed-Solomon code. Cryptography and Communications, 16(4): 857-878, 2024.

\bibitem{}\label{t21}
Huang D, Yue Q, Niu Y, et al. MDS or NMDS self-dual codes from twisted generalized Reed-Solomon codes. Designs, Codes and Cryptography, 89(9): 2195-2209, 2021.

\bibitem{}\label{t22}
Sui J, Yue Q, Sun F. New constructions for self-dual codes via twisted generalized ReedSolomon codes. Cryptography and Communications, 15(5): 959-978, 2023.

\bibitem{}\label{t13}
Liu H, Liu S. Construction of MDS twisted Reed-Solomon codes and LCD MDS codes. Designs, Codes and Cryptography, 89(9): 2051-2065, 2021.

\bibitem{}\label{t11}
Carlet C, Mesnager S, Tang C, Qi Y, Pellikaan R. Linear codes over Fq are equivalent to LCD codes for $q>3$. IEEE Transactions on Information Theory, 64(4): 3010-3017, 2018.

\bibitem{}\label{t12}
Li C, Ding C, Li S. LCD cyclic codes over finite fields. IEEE Transactions on Information Theory, 63(7): 4344-4356, 2017.

\bibitem{}\label{t14}
Wu Y, Lee Y. Binary LCD codes and self-orthogonal codes via simplicial complexes. IEEE Communications Letters, 24(6): 1159-1162, 2020.


\bibitem{}\label{t24}
Wu Y, Hyun J Y, Lee Y J. New LCD MDS codes of non-Reed-Solomon type. IEEE Transactions on Information Theory, 67(8): 5069-5078, 2021.

\bibitem{}\label{t25}
Beelen P, Bossert M, Puchinger S, Rosenkilde J. Structural properties of twisted Reed-Solomon codes with applications to cryptography. 2018 IEEE International Symposium on Information Theory, : 946-950, 2018.

\bibitem{}\label{t27}
Zhang J, Zhou Z, Tang C. A class of twisted generalized Reed-Solomon codes. Designs, Codes and Cryptography, 90(7): 1649-1658, 2022.

\bibitem{}\label{t28}
Sui J, Yue Q, Li X, Huang D. MDS, near-MDS or 2-MDS selfdual codes via twisted generalized Reed-Solomon codes. IEEE Transactions on Information Theory, 68(12): 7832-7841, 2022.

\bibitem{}\label{t29}
Singh H, Meena K C. MDS multi-twisted Reed-Solomon codes with small dimensional hull. Cryptography and Communications, 16(3): 557-578, 2024.

\bibitem{}\label{t30}
Yang S, Wang J, Wu Y. Two classes of twisted generalized Reed-Solomon codes with two twists. Finite Fields and Their Applications, 104: 102595, 2025.

\bibitem{}\label{t32}
Meena K C, Pachauri P, Awasthi A, Bhaintwal M. A class of triple-twisted GRS codes. Designs, Codes and Cryptography, 93(7): 2369-2393, 2025.

\bibitem{}\label{t33}
Gu H, Zhang J. On twisted generalized Reed-Solomon codes with l twists. IEEE Transactions on Information Theory, 70(1): 145-153, 2024.

\bibitem{}\label{t34}
Cheng W. On parity check matrices of twisted generalized Reed-Solomon codes. IEEE Transactions on Information Theory, 70(5): 3213-3225, 2024.

\bibitem{}\label{t35}
Ding Y, Zhu S. New self-dual codes from TGRS codes with general l twists. Advances in Mathematics of Communications, 19(2): 662-675, 2025.

\bibitem{}\label{t31}
Zhao C, Ma W, Yan T, Sun Y. Research on the construction of maximum distance separable codes via arbitrary twisted generalized Reed-Solomon codes. IEEE Transactions on Information Theory, 71(7): 5130-5143, 2025.

\bibitem{}\label{t37}
Hu Z, Wang L, Li N, et al. On $(\mathcal{L},\mathcal{P})$-twisted generalized Reed-Solomon codes. IEEE Transactions on Information Theory, : , 2025.

\bibitem{}\label{t41}
Liang Z, Liao Q. Four classes of LCD codes from $(\ast)-(\mathcal{L}, \mathcal{P})$-twisted generalized Reed-Solomon codes. arXiv:2509.14878v1, 2025.



\bibitem{}\label{t43}
Dodunekov S M, Landjev I N. Near-MDS codes over some small fields. Discrete Mathematics, 213(1-3): 55-65, 2000.

\bibitem{}\label{t44}
Abdukhalikov K, Ding C, Verma G K. Some constructions of non-generalized ReedSolomon MDS codes. arXiv:2506.04080, 2025.

\bibitem{}\label{t45}
Macdonald I G. Symmetric Functions and Hall Polynomials. Oxford University Press, Oxford, 2nd edition, 1995.

\bibitem{}\label{t46}
Ding J, Zhou A. Eigenvalues of rank-one updated matrices with some applications. Applied Mathematics Letters, 20(12): 1223-1226, 2007.

\bibitem{}\label{t47}
Liang, Z., Jia, C., Huang, D., Liao, Q., Tang, C. Multi-twisted generalized reed-
solomon codes: Structure, properties, and constructions, 2025.

















%
%
%
%
%
%
%
%
%
%
%
%
%
%
%
%
%
%
%
%
%
%
%
%
%
%
%
%
%
%
%
%
%
%
%
%
%
%
%
%
%
%
%
%
%
%
%
%
%
%
%
%
%
%
%
%
%
%


\end{thebibliography}
\end{document}